\documentclass[12pt]{article}
\usepackage{mathrsfs}
\usepackage{dsfont}
\usepackage{amsthm}
\usepackage{mathrsfs}
\usepackage{amsmath}
\usepackage{amsfonts}
\usepackage[colorlinks,linkcolor=black,anchorcolor=blue,citecolor=blue]{hyperref}
\usepackage{amssymb, amsmath, cite}

\newtheorem{theorem}{Theorem}[section]
\newtheorem{lemma}{Lemma}[section]

\newcommand\be{\begin{equation}}
\newcommand\ee{\end{equation}}
\newcommand\ber{\begin{eqnarray}}
\newcommand\eer{\end{eqnarray}}
\newcommand\berr{\begin{eqnarray*}}
\newcommand\eerr{\end{eqnarray*}}
\newcommand\bea{\begin{eqnarray}}
\newcommand\eea{\end{eqnarray}}

\newcommand\vp{\varphi}\newcommand{\nn}{\nonumber}
\newcommand\vep{\varepsilon}\newcommand{\ii}{\mathrm{i}}
\newcommand{\dd}{\mathrm{d}}
\newcommand\e{\mathrm{e}}\newcommand\pa{\partial}
{\setlength\arraycolsep{2pt}

\begin{document}

\title{Multiple cosmic strings in Chern--Simons--Higgs theory with gravity}
\author{Lei Cao, Shouxin Chen\\School of Mathematics and Statistics\\
Henan University, Kaifeng, Henan 475004, PR China}
\date{}
\maketitle

\noindent{\bf abstract.} In this paper, we consider the self--dual equation arising from Abelian Chern--Simons--Higgs theory coupled to the Einstein equations over the plane $\mathbb{R}^2$ and a compact surface $S$. We prove the existence of symmetric topological solutions and non--topological solutions on the plane by using the fixed--point theorem and a shooting method, respectively. A necessary and sufficient condition related to the string number $N$, the Euler characteristic $\chi(S)$ of $S$, and the gravitational coupling factor $G$ is given to show the existence of $N$--string solutions over a compact surface.

\medskip


\medskip


\section{Introduction}\label{s1}
\setcounter{equation}{0}

The Lagrangian density of the Abelian Chern--Simons--Higgs theory of Contatto \cite{Con}, Izquierdo, Fuertes and Guilarte \cite{Izq} is given by
\be\label{1.1}
\mathcal{L}=-\frac{\kappa^2}{16|u|^2}g^{\mu\mu'}g^{\nu\nu'}F_{\mu\nu}F_{\mu'\nu'}+\frac{1}{2}g^{\mu\nu}D_\mu u\overline{D_\nu u}-\frac{1}{2\kappa^2}|u|^2(1-|u|^2)^2,
\ee
where $\kappa$ is a nonzero real constant, $u$ is a complex scalar field which can be viewed as a Higgs field, $g_{\mu\nu}=\text{diag}(1,-1,-1,-1)$ is the metric of the Minkowski spacetime, $F_{\mu\nu}=\partial_\mu A_\nu-\partial_\nu A_\mu$ is the electromagnetic curvature induced from a real-valued gauge vector field $A_\mu~(\mu, \nu=0,1,2,3, t=x_0)$, $D_\mu u=\partial_\mu u-\ii A_\mu u$ is the gauge covariant derivative. The associated energy momentum tensor reads
\be\label{1.2}
T_{\mu\nu}=-\frac{\kappa^2}{4|u|^2}g^{\mu'\nu'}F_{\mu\mu'}F_{\nu\nu'}+\frac{1}{2}(D_\mu u\overline{D_\nu u}+\overline{D_\mu u}D_\nu u)-g_{\mu\nu}\mathcal{L}.
\ee

The self--dual Chern--Simons--Higgs theory \cite{Ho,Jac} has been an important topic in physics to study multiply distributed electrically and magnetically charged vortices \cite{Ca,Cha,S,Sp,T,YaS,Chan}, which are soliton--like structures and occur as a consequence of spontaneous symmetry breaking in internal spaces and are often characterized as topological defects. When gravity is taken into account, the vortices, which appear as clumps of energy, clearly cause the spacetime geometry to inherit such properties, as demonstrated by Einstein equations,
\be\label{1.3}
G_{\mu\nu}=-8\pi G T_{\mu\nu},
\ee
where $G$ is the Newton's gravitational constant, $G_{\mu\nu}$ is the Einstein tensor, and $T_{\mu\nu}$ represents the energy--momentum tensor of the matter sector. More precisely, it is natural to expect the geometry of spacetime to "curl up" around such clumps of energy, which would provide a possible mechanism for the sites to appear as seeds for matter accretion or accumulation in the early universe. In particular, vortex lines have been used in cosmology to generate large--scale string structures, known as cosmic strings \cite{Wi,Y,Yang,G,Gr,K,V,Vi,T,YaS}.

If the gravitational metric element takes the form \cite{Co}
\be\label{1.4}
\dd s^2=\dd t^2-(\dd x^3)^2-\e^\eta\left((\dd x^1)^2+(\dd x^2)^2\right),~~\eta=\eta(x^1, x^2).
\ee
We see the Einstein equations \eqref{1.3} are recast into the following single equation
\be\label{1.7}
K_\eta=8\pi G\mathcal{H}
\ee
over $(\mathbb{R}^2,\e^\eta\delta_{ij})$, where the Gauss curvature is now denoted as $K_\eta$ to emphasize its dependence on the conformal exponent $\eta$, which is given by
\be\label{1.5}
K_\eta=-\frac 1 2\e^{-\eta}\triangle\eta,
\ee
and $\mathcal{H}=T_{00}$ is recognized as the Hamiltonian energy density.

We now use the string metric \eqref{1.4} and stay within the static and axial symmetric situation, namely the complex scalar field $u$ depends only on $x^1, x^2$. Then, the Hamiltonian energy density is
\bea\label{1.8}
\mathcal{H}&=&T_{00}=-T_{33}=-\mathcal{L}\nn\\
&=&\frac{\kappa^2}{8|u|^2}\e^{-2\eta}F_{12}^2+\frac 12\e^{-\eta}(|D_1 u|^2+|D_2 u|^2)+\frac{|u|^2(1-|u|^2)^2}{2\kappa^2}\nn\\
&=&\frac 12\left(\left(\e^{-\eta}\frac{\kappa}{2|u|}F_{12}\mp \frac{|u|(1-|u|^2)}{\kappa}\right)^2+\e^{-\eta}|D_1 u\pm \ii D_2 u|^2\right)\pm \frac 12\e^{-\eta}F_{12}\nn\\
&& \pm \frac 12\e^{-\eta}\left(-|u|^2F_{12}+\ii(D_1 u \overline{D_2 u}-\overline{D_1 u} D_2 u)\right).
\eea
The terms on the right-hand side of \eqref{1.8} are grouped into three quantities such that the first term consists of quadratures, the second one is topological, and the last term may be recognized as the total divergence of a current density. For this purpose, inspired by \cite{Co}, we form the current density
\be\label{1.9}
J_i=\frac{\ii}{2}(u \overline{D_i u}-\overline{u} D_i u),~~i=1,2.
\ee
Then the right-hand side of \eqref{1.8} becomes
\be\label{1.10}
\mathcal{H}=\frac 12\left(\left(\e^{-\eta}\frac{\kappa}{2|u|}F_{12}\mp \frac{|u|(1-|u|^2)}{\kappa}\right)^2+\e^{-\eta}|D_1 u\pm \ii D_2 u|^2\right)\pm \frac 12\e^{-\eta} F_{12}\pm \frac 12\e^{-\eta} J_{12},
\ee
where $J_{12}=\pa_1 J_2-\pa_2 J_1$.

Furthermore, integrating \eqref{1.10} over $(\mathbb{R}^2,\e^\eta\delta_{ij})$, we derive the Bogomol'nyi topological lower bound $E=\int \mathcal{H}\e^\eta\dd x\geq \pi N$ where $N$ is the total string number and this lower bound is saturated if and only if the two quadratic terms in \eqref{1.8} identically vanish:
\bea
\frac{\kappa}{2|u|}F_{12}&=&\pm \e^\eta \frac{|u|(1-|u|^2)}{\kappa},\label{1.11}\\
D_1 u\pm \ii D_2 u&=&0.\label{1.12}
\eea
To preserve the finiteness of energy, it is necessary that one of the following boundary conditions be satisfied
\bea
u\to &0& ~~\text{as}~~|x|\to\infty,\label{1.13}\\
|u|^2\to &1& ~~\text{as}~~|x|\to\infty.\label{1.14}
\eea
The boundary condition \eqref{1.13} is called non--topological, and \eqref{1.14} is called topological.

To proceed, we know from \eqref{1.12} and the $\overline{\pa}$-Poincar\'{e} lemma \cite{Jaf} that the zeros of $u$ are all discrete and of integer multiplicities. Let the zeros of $u$ be $p_1,p_2,\ldots,p_N$ (counting multiplicities). Thus, setting $v=\ln|u|^2$, we obtain from \eqref{1.11} and \eqref{1.12} the governing equation
\be\label{1.15}
\triangle v=4\e^\eta\frac{\e^v(\e^v-1)}{\kappa^2}+4\pi\sum_{s=1}^N \delta_{p_s}(x),~~x\in\mathbb{R}^2.
\ee

Now we resolve the Einstein equations \eqref{1.7} to obtain the conformal factor $\e^\eta$ in terms of the unknown $v$ in \eqref{1.15}. To this end, using \eqref{1.11} to represent $\mathcal{H}$ given in the second line of \eqref{1.8} as
\be\label{1.16}
\mathcal{H}=\pm\frac 12\e^{-\eta}F_{12}(1-|u|^2)+\frac 12\e^{-\eta}(|D_1 u|^2+|D_2 u|^2).
\ee
When stay away from the zeros of $u$ and use the relation $|u|^2=\e^v$, we may recast \eqref{1.16} as
\be\label{1.17}
4\e^\eta\mathcal{H}=(\e^v-1)\triangle v+\e^v|\nabla v|^2.
\ee
Taking account of the zeros $p_1,p_2,\ldots,p_N$ of $u$, \eqref{1.17} leads us to
\be\label{1.18}
4\e^\eta\mathcal{H}=\triangle(\e^v-v)+4\pi\sum_{s=1}^N \delta_{p_s}(x),~~x\in\mathbb{R}^2.
\ee
With this observation and returning to \eqref{1.7} with \eqref{1.5}, we conclude that the quantity
\be\label{1.19}
\frac{\eta}{4\pi G}+\e^v-v+\sum_{s=1}^N\ln|x-p_s|^2
\ee
is a harmonic function over $\mathbb{R}^2$, which may be taken to be an arbitrary constant. Therefore, we obtain the gravitational metric factor
\be\label{1.20}
\e^\eta=\lambda\left(\e^{v-\e^v}\prod_{s=1}^N|x-p_s|^{-2}\right)^{4\pi G},~~\lambda>0,
\ee
where $\lambda>0$ is a free coupling parameter which may be taken to be arbitrarily large.

Inserting \eqref{1.20} into \eqref{1.15}, we are lead to the governing equation for $N$ prescribed cosmic strings
\be\label{1.21}
\triangle v=\beta\left(\e^{v-\e^v}\prod_{s=1}^N|x-p_s|^{-2}\right)^a\e^v(\e^v-1)+4\pi\sum_{s=1}^N \delta_{p_s}(x),~~x\in\mathbb{R}^2,
\ee
where we set $\beta=\frac{4\lambda}{\kappa^2}, a=4\pi G$.

The purpose of this paper is to show the existence of solutions of \eqref{1.21}. Because of the singularity of the gravitational term $\e^\eta$ at $p_1,p_2,\ldots,p_N$ and the the non--monotonicity of $\e^v(\e^v-1)$ with respect to $v$. The topological solutions of \eqref{1.21} are difficulty to obtain. Besides, the energy and charges of a non--topological solution depend on its accurate decay rate at infinity which is also quite hard to get. Thus, we only present a complete study of radially symmetric solutions of \eqref{1.21} on the plane. Our main results concerning topological and non--topological solutions of \eqref{1.21} are as follows.

\begin{theorem}\label{th2.1}
Under the condition $4\pi G N=1$ and $N\geq1$, if all the points $p_s$ are identical, that is to say, $p_s=p_0, s=1,2,\cdots, N$. Then the equation \eqref{1.21} has a symmetric topological solution that vanishes at infinity with the choice
\be\nn
\beta=\frac{2 N^2}{\int_{-\infty}^0 \e^{4\pi G(v-\e^v)}\e^v(1-\e^v)\dd v},
\ee
and possess the sharp decay behavior
\be\nn
|v(r)|=O(r^{-(1-\vep)\sqrt{\beta\e^{-a}}}),~~|v_r(r)|=O(r^{-(1-\vep)\sqrt{\beta\e^{-a}}-1}),~~r=|x|\to\infty.
\ee
\end{theorem}

The conditions stated in Theorem \ref{th2.1} are reasonable. In fact, $N$ denotes the number of $u$ zeros, so it must be positive. The formulation of $\e^\eta$ \eqref{1.20} gives us the estimates in $\mathbb{R}^2$,
\be\label{1.24}
C_1(1+|x|)^{-8\pi GN}\leq\e^{\eta(x)}\leq C_2(1+|x|)^{-8\pi GN},
\ee
where $C_1,C_2>0$ are suitable constants. With regard to the classical Hopf--Rinow--de Rham theorem, it is straightforward to see from \eqref{1.24} that a solution leads to a geodesically complete metric if and only if the total string number $N$ satisfies the upper bound $N\leq 1/4\pi G$, which implies that the conditions given in Theorem \ref{th2.1} are appropriate.

\begin{theorem}\label{th2.2}
For any $\tilde{x}\in\mathbb{R}$, $4\pi G N<1$, $N\geq1$ and any
$$
\alpha\geq\ln\frac{4\pi G}{1+4\pi G-\sqrt{1+4\pi G}},
$$
the equation \eqref{1.21} has a symmetric non--topological solution $v$ which tends to negative infinity at infinity so that $v$ have a global maximum at $r=|x|=|\tilde{x}|$ and
\be\nn
\max_{r\in [0,\infty)}v(r)=-\alpha,~~\lim_{r\to\infty}r v_r(r)=-k,
\ee
where $k>4+2N-16\pi GN$ is a constant which may depend on $\alpha$ and $\tilde{x}$. In fact, $k$ may taking any value in the open interval $(4+2N-16\pi GN,\infty)$.
\end{theorem}

It is well known that the Euler characteristic satisfies the equation $\chi(S)=2-2n$, where $n$ is called the genus of a compact orientable 2--surface $S$. On the other hand, we call a solution of \eqref{1.21} over a compact Riemann surface $S$ with first Chern number equal to $N$ and $N$--string if and only if there holds the exact relation
\be\label{t3.3a}
\chi(S)=4\pi GN.
\ee
This property implies that the only possible situation we may have so that the equation \eqref{1.21} induced from the Einstein and Abelian Higgs system has a cosmic string solution is given by $n=0$. In other words, the 2--surface must be diffeomorphic to the Riemann sphere $S^2$ and all other geometries with $n\neq 0$ are excluded. In particular, the surface $S$ can not be a tours ($n=1$). This result implies that gravitational string condensation, realized by the appearance of a periodic lattice structure, fails to exist. Our result of the existence of multiple strings over a compact surface reads as follows.

\begin{theorem}\label{th2.3}
Given an integer $N\geq 3$, there is an $N$--string solution for the equation \eqref{1.21} over an appropriate compact Riemann surface $S$, only if the string number $N$, the Euler characteristic $\chi(S)$ of $S$, and the gravitational coupling factor $G$ satisfy the exact relation
$4\pi GN=2$.
\end{theorem}

The rest of this paper is organized as follows. In Section 2, we prove the existence of the symmetric topological solutions via a fixed point theorem in the first subsection. In the second subsection, we establish the symmetric non--topological solutions through a shooting analysis. In Section 3, we provide the proof of the existence of an $N$--string solution over a compact surface by using a monotone iteration method.

\section{Construction of symmetric solutions}\label{s2}
\setcounter{equation}{0}

In this section, we consider the radically symmetric solutions of \eqref{1.21}. For this purpose, we assume that all the zeros $p_s$'s are identical: $p_s=p_0, s=1,2,\ldots, N$. Without loss of generality, let $p_s=p_0=\text{the origin of}~\mathbb{R}^2, s=1,2,\ldots, N$. Use $r=|x|$, then we arrive at the reduced form of the elliptic equation \eqref{1.21}
\be\label{2.1}
\triangle v=\beta r^{-2aN}\e^{a(v-e^v)}\e^v(\e^v-1)+4\pi N\delta(x).
\ee
According to the removable singularity theorem, we come up with an equivalent form of the equation \eqref{2.1},
\bea
(rv_r)_r&=&\beta r^{1-2aN}\e^{a(v-e^v)}\e^v(\e^v-1),~~r>0,\label{2.2}\\
\lim_{r\to 0}\frac{v(r)}{\ln r}&=&\lim_{r\to 0}r v_r(r)=2N.\label{2.3}
\eea

In the following, we will use a fixed--point theorem developed in \cite{Yang} to show the existence of symmetric topological solutions and a shooting method developed in \cite{YaS, Cao} to establish the existence of symmetric non--topological solutions of \eqref{2.2}--\eqref{2.3}.

\subsection{Symmetric topological solutions}

We now introduce the new variable $t=\ln r$ and use the condition $aN=1$ to rewrite the system \eqref{2.2}--\eqref{2.3} as
\bea
v''&=&\beta\e^{a(v-e^v)}\e^v(\e^v-1),~~-\infty<t<\infty,\label{2.4}\\
\lim_{t\to -\infty}\frac{v(t)}{t}&=&\lim_{t\to -\infty}v'(t)=2N,\label{2.5}
\eea
where $v'$ is the derivative of $v$ with respect to $t$. Using the maximum principle, we see that if $v$ is a topological solution of \eqref{2.4}--\eqref{2.5} then there must hold $v(t)\leq0$ for all $t\in (-\infty, \infty)$. Denote the right--hand side of the equation \eqref{2.4} as $h(v)$, then the system \eqref{2.4}--\eqref{2.5} equals to
\be\label{2.6}
v(t)=2Nt+\int_{-\infty}^t(t-\tau)h(v(\tau))\dd \tau,~~~t\in\mathbb{R}.
\ee
For convenient, set $w=v-2Nt$, so we can rewrite \eqref{2.6} as
\be\label{2.7}
w(t)=\int_{-\infty}^t(t-\tau)h(2N\tau+w(\tau))\dd \tau,~~~t\in\mathbb{R}.
\ee
We will look for a solution of \eqref{2.7} satisfying $w(t)\to 0, t=-\infty$. For this purpose, denote the right--hand side of the equation \eqref{2.7} as $T(w)$, then we arrive at a fixed--point problem, $w=T(w)$. Define the following function space
\be\nn
\mathcal{X}=\left\{w\in C(-\infty, t_0]~\Big|~\lim_{t\to -\infty}w(t)=0,~~\sup_{t\leq t_0}|w(t)|\leq 1\right\},
\ee
where $-\infty<t_0<\infty$ is to be determined later. First, we show that $T$ maps from $\mathcal{X}$ into itself for some proper $t_0$. In fact, according to the definition of $h$, for any $w\in\mathcal{X}$, we have
\be\label{2.8}
|T(w)|\leq\beta\sup_{t\leq t_0}\int_{-\infty}^t|t-\tau|\e^{a(2N\tau+w(\tau)-\e^{2N\tau+w(\tau)})}\e^{2N\tau+w(\tau)}(\e^{2N\tau+w(\tau)}-1)\dd \tau \leq 1
\ee
holds for some suitable $t_0$.
Next, we show that $T$ is a contraction operator. To achieve this goal, differentiating the right--hand side of \eqref{2.4}, we get
\be\label{2.9}
|h'(v)|=\beta\e^{a(v-e^v)}\e^v\left|2\e^v-1-a(1-\e^v)^2\right|\leq C_1\e^{av},
\ee
where $C_1=C_1(\beta)>0$. Hence, for any $w_1, w_2\in\mathcal{X}$, we have
\bea\label{2.10}
&&|T(w_1)-T(w_2)|\nn\\
=&&\left|\int_{-\infty}^t(t-\tau)h'\left(2N\tau+\widetilde{w}(\tau)\right)(w_1(\tau)-w_2(\tau))\dd \tau\right|\nn\\
\leq && C_1\sup_{t\leq t_0}|w_1(t)-w_2(t)|\int_{-\infty}^t(t_0-\tau)\e^{a(2N\tau+1)}\dd \tau,
\eea
where $\widetilde{w}(t)$ lies between $w_1(t)$ and $w_2(t)$. Therefore, when $t_0$ is properly chosen, the inequality \eqref{2.10} implies $T: \mathcal{X}\to\mathcal{X}$ is a contraction operator. Consequently, the above argument says that \eqref{2.7} has a unique solution in the neighborhood of $t=-\infty$. Besides, taking the derivative of \eqref{2.7}, we get $w'(t)\to 0$ as $t\to -\infty$ which means the boundary condition \eqref{2.5} is provided. Furthermore, we can extend $w$ to a solution over the entire $\mathbb{R}$ by using \eqref{2.7} and a standard continuation argument. Therefore, the system \eqref{2.4}--\eqref{2.5} is solved. While, in order to find a topological solution of \eqref{2.1}, the boundary condition
\be\label{2.11}
\lim_{t\to\infty}v(t)=0
\ee
remains to be achieved. To this end, define
\be\nn
H(v)=-\int_{-\infty}^v h(w)\dd w>0,
\ee
then multiply the both side of \eqref{2.4} by $v'$ and integrating over $(-\infty, t)$, we have
\be\label{2.12}
(v'(t))^2=4N^2-2 H(v(t)).
\ee
It is useful to study the critical points of the equation \eqref{2.12} first. Suppose $\tilde{v}$ satisfies
\be\label{2.13}
2N^2-H(\tilde{v})=2N^2+\int_{-\infty}^{\tilde{v}}h(w)\dd w=0.
\ee
In order to ensure uniqueness at the equilibrium $\tilde{v}$, we need to require that
\be\label{2.14}
H'(\tilde{v})=-h(\tilde{v})=\beta \e^{a(\tilde{v}-\e^{\tilde{v}})}\e^{\tilde{v}}(1-\e^{\tilde{v}})=0.
\ee
Thus the only choice is  $\tilde{v}=0$. Inserting this result into \eqref{2.13}, we can determine the parameter $\beta$,
\be\label{2.15}
\beta=\frac{2 N^2}{\int_{-\infty}^0 \e^{a(v-\e^v)}\e^v(1-\e^v)\dd v}.
\ee
Although we represent $\beta$ in an integral form without giving a concrete value, we can see from the inequality
\be
0<\int_{-\infty}^0\e^{a(v-\e^v)}\e^v(1-\e^v)\dd v
\leq \int_{-\infty}^0\e^{a (v-\e^v)}(1-\e^v)\dd v=\frac{1}{a\e^a}\nn
\ee
that $\beta$ is a finite number depending only on $a$.

Since $v'(-\infty)>0$, we can rewrite \eqref{2.12} as
\be\label{2.16}
v'(t)=\sqrt{4N^2-2H(v)}\equiv\sqrt{F(v)}
\ee
and $v(t)$ is an increasing function. Besides, there holds the limit
\be\label{2.17}
\lim_{v\to 0^-}\frac{\sqrt{F(v)}}{v}=-\sqrt{\frac{F''(0)}{2}}=-\sqrt{\beta\e^{-a}}.
\ee
Furthermore, \eqref{2.16} can be rewritten in the integral form
\be\label{2.18}
\int_{v(\tau)}^{v(t)}\frac{\dd v}{\sqrt{F(v)}}=t-\tau.
\ee
With the above properties \eqref{2.17} and \eqref{2.18}, we are ready to show that $v$ satisfies the desired condition \eqref{2.11}. In fact, since $F(v)$ decreases monotonically with respect to $v\leq 0$ and $F(0)=0$, we have $F(v)>0$ as $v<0$. Furthermore, in view of \eqref{2.16}, we see that $v'(t)>0$ whenever $v<0$. In the following we will show that $v(t)<0$ for all $t$. Otherwise, if there exists some $t_0\in(-\infty, \infty)$ such that $v(t_0)=0$. Then the limit \eqref{2.17} says that there is a number $\delta>0$, such that
\be\label{2.19}
\frac{\sqrt{F(v)}}{v}<-\frac{\sqrt{\beta\e^{-a}}}{2},~~t_0-\delta\leq t<t_0.
\ee
Note \eqref{2.18} and \eqref{2.19}, we arrive at
\be\nn
t-(t_0-\delta)>-\frac{2}{\sqrt{\beta\e^{-a}}}\int_{v(t_0-\delta)}^{v(t)}\frac{\dd v}{v}=\frac{2}{\sqrt{\beta \e^{-a}}}\ln\Big|\frac{v(t_0-\delta)}{v(t)}\Big|,~~t_0-\delta<t<t_0,
\ee
whose left--hand side is finite while the right--hand side approaches infinity as $t\to t_0$. This is a contradiction. Hence, $v(t)<0$ and $v'(t)>0$ for all $t\in(-\infty, \infty)$. Particularly, the limit of $v(t)$ as $t\to\infty$ exists. We claim that the only possible value of this limit is $v_\infty=0$. Otherwise, assume $v_\infty<0$, then $F(v)$ is bounded from blow by $F(v_\infty)>0$. We see that as $t\to\infty$ the left--hand side of \eqref{2.18} remains bounded while the right--hand side of \eqref{2.18} approaches infinity. Consequently, we have $v(t)\to 0$ as $t\to\infty$.

To proceed with asymptotic of the symmetric topological solution of \eqref{2.1}. For any $0<\vep<1$, let $\tau>0$ be large enough so that
\be\label{2.20}
\frac{\sqrt{F(v(t))}}{v(t)}<-(1-\vep)\sqrt{\beta\e^{-a}},~~t\geq\tau.
\ee
Inserting \eqref{2.20} into \eqref{2.18}, we get
\be\nn
t-\tau>-\frac{1}{(1-\vep)\sqrt{\beta\e^{-a}}}\int_{v(\tau)}^{v(t)}\frac{\dd v}{v}=\frac{1}{(1-\vep)\sqrt{\beta\e^{-a}}}\ln\bigg|\frac{v(\tau)}{v(t)}\bigg|,~~t>\tau,
\ee
which leads us to arrive at the decay estimate
\be\label{2.21}
|v|=O(\e^{-(1-\vep)\sqrt{\beta\e^{-a}}t}),~~~t\to\infty.
\ee
Returning to the original variable $r=\e^t$, we obtain that
\be\label{2.22}
|v(r)|=O(r^{-(1-\vep)\sqrt{\beta\e^{-a}}}),~~~r\to\infty.
\ee
Note that $\frac{v'(t)}{|v(t)|}\to\sqrt{\beta\e^{-a}}$ as $t\to\infty$,
we have
\be\label{2.23}
|v_r(r)|=O(r^{-(1-\vep)\sqrt{\beta\e^{-a}}-1}),~~~r\to\infty.
\ee
The theorem \ref{th2.1} is thus proven.

\subsection{Symmetric non--topological solutions}

In this subsection, we present a complete study of radical symmetric solutions with non--topological boundary condition. We first formulate the problem and then solve the problem through a shooting analysis. The non--topological solutions of \eqref{2.2}--\eqref{2.3} need to following the boundary condition at $r=\infty$,
\be\label{4.1}
\lim_{r\to\infty}v(r)=-\infty,
\ee
which implies that we need to solve a two--point boundary value problem over the infinite interval $(0, \infty)$. In the following, we shall use a two--side shooting analysis to solve the problem. In other words, we are going to show that for a suitable $r_0>0$, there are global solutions of \eqref{2.2} coupled with some adequate initial data at $r=r_0$ to fulfill \eqref{2.3} and \eqref{4.1}.

We first derive an a priori estimate for the solutions of \eqref{2.2}.

\begin{lemma}\label{le4.1}
If $v(r)$ is a solution of \eqref{2.2} satisfying
\be\nn
\lim_{r\to 0}v(r)=-\infty,~~~\lim_{r\to\infty}v(r)=-\infty,
\ee
then $v(r)<0$ for all $r>0$.
\end{lemma}

\begin{proof}
We can reach the conclusion directly through a maximum principle.
\end{proof}

Lemma \ref{le4.1} tells us that the solution of \eqref{2.2} under the boundary conditions \eqref{2.3} and \eqref{4.1} must have a global maximum $v_0=-\alpha$ at some $r=r_0>0$. Therefore we are going to find solutions of \eqref{2.2} under the initial condition
\be\label{4.2}
v(r_0)=-\alpha,~~~v_r(r_0)=0.
\ee

For simplicity, we introduce a change of independent variable
\be\label{4.2b}
t=\ln r,~~~t_0=\ln r_0.
\ee
Then \eqref{2.2} and \eqref{4.2} become
\bea
v''&=&\beta \e^{(2-2aN)t}\e^{a(v-\e^v)}\e^v(\e^v-1),~~~-\infty<t<\infty,\label{4.3}\\
v(t_0)&=&-\alpha,~~~v'(t_0)=0,\label{4.4}
\eea
where, and in the sequel of this section, $v'=\frac{\dd v}{\dd t}$ and $v(t)$ denotes the dependence of the solution $v$ of \eqref{2.2} on the new variable $t$.

\begin{lemma}\label{le4.2}
For any $t_0\in \mathbb{R}$, $\alpha>0$ and $aN<1$, there exists a unique global solution $v(t)$ of \eqref{4.3}--\eqref{4.4}. This solution satisfies $v(t)<0$ and
\be\nn
\lim_{t\to -\infty}v(t)=-\infty,~~~\lim_{t\to\infty}v(t)=-\infty.
\ee
\end{lemma}

\begin{proof}
Let $v(t)$ be a local solution of \eqref{4.3}--\eqref{4.4}, then in the interval of existence, there holds
\be\label{4.5}
v'(t)=\beta\int_{t_0}^t \e^{(2-2aN)s}\e^{a(v(s)-\e^{v(s)})}\e^{v(s)}(\e^{v(s)}-1)\dd s.
\ee

We first show that, for all $t$, where $v(t)$ exists, there holds $v(t)<0$. Otherwise, if there exists a $\tilde{t}>t_0$ such that $v(\tilde{t})\geq 0$, we may assume $\tilde{t}$ is such that
\be\nn
\tilde{t}=\inf\{t\geq t_0~|~v(t)~\text{exists and}~v(t)\geq 0\}.
\ee
Then $\tilde{t}>t_0$ and $v(\tilde{t})=0$. It is obvious that $v(t)<0$ for all $t_0\leq t<\tilde{t}$. However, \eqref{4.5} implies $v'(t)<0$ for all $t_0<t\leq\tilde{t}$. Thus $v(\tilde{t})<0$, which makes a contradiction.

Similarly, if there exists a $\tilde{t}<t_0$ such that $v(\tilde{t})\geq 0$ and $v(t)<0$ for all $\tilde{t}<t\leq t_0$. Then $v'(t)>0$ for $\tilde{t}\leq t<t_0$. So $v(\tilde{t})<0$, which makes a contradiction again.

In view of $v(t)<0$ and \eqref{4.5}, we see that $v'(t)$ cannot blow up in finite time. Therefore, the solution of \eqref{4.3}--\eqref{4.4} exists globally in $t\in(-\infty, \infty)$.

Second, we show that $v(t)\to -\infty$ as $t\to -\infty$. From the property $v(t)<0$ and \eqref{4.5}, we see
\be\nn
\lim_{t\to-\infty}v'(t)=\beta\int_{t_0}^{-\infty} \e^{(2-2aN)s}\e^{a(v(s)-\e^{v(s)})}\e^{v(s)}(\e^{v(s)}-1)\dd s=C>0,
\ee
which implies $v(t)\to -\infty$ as $t\to-\infty$.

Now, we show the behavior $v(t)\to -\infty$ as $t\to\infty$. Because of $v(t)<0$ and \eqref{4.5}, we have $v'(t)<0$ for $t>t_0$. Thus, either $v(t)\to-\infty$ or $v(t)\to$ a finite number $a<-\alpha<0$ as $t\to\infty$. However, only the former case can happen. Otherwise, assume $a<v(t)\leq-\alpha$, we obtain from \eqref{4.5} that
\bea\label{4.6}
v'(t)&\geq&\beta\int_{t_0}^t \e^{(2-2aN)s}\left[\min_{a\leq v\leq -\alpha}\left\{\e^{a(v(s)-\e^{v(s)})}\e^{v(s)}(\e^{v(s)}-1)\right\}\right]\dd s\nn\\
&=&-C_1\int_{t_0}^t \e^{(2-2aN)s}\dd s=\frac{C_1}{2aN-2}(\e^{(2-2aN)t}-\e^{(2-2aN)t_0}),~~t>t_0,
\eea
where $C_1>0$ is a constant. As a consequence of \eqref{4.6}, we see that $v(t)\to -\infty$ as $t\to\infty$. This reaches a contradiction.

Finally, we see from the uniqueness theorem for solutions to initial value problems of ordinary differential equations that the global solution $v(t)$ is unique.
\end{proof}

With respect to the new variable $t$ defined in \eqref{4.2b}, the boundary condition \eqref{2.3} reads
\be\label{4.7}
\lim_{t\to-\infty}\frac{v(t)}{t}=2N.
\ee

\begin{lemma}\label{le4.3}
For any given $\alpha\geq\ln\frac{a}{1+a-\sqrt{1+a}}$, there exists $t_0=t_0(\alpha)$ such that the unique solution of \eqref{4.3}--\eqref{4.4} satisfies the condition \eqref{4.7}.
\end{lemma}

\begin{proof}
For any $t_0\in(-\infty, \infty)$ and $\alpha>0$, let $v=v(t; t_0, \alpha)$ be the unique global solution of \eqref{4.3}--\eqref{4.4}. Then $v<0$ and $v\to -\infty$ as $t\to-\infty$ by Lemma \ref{le4.2}. Therefore, we see from \eqref{4.7} that
\be\label{4.8}
\eta(t_0, \alpha)\equiv\lim_{t\to-\infty}v'(t; t_0, \alpha)=2N.
\ee
Furthermore, \eqref{4.5} gives us
\be\label{4.9}
\eta(t_0, \alpha)=-\beta\int_{-\infty}^{t_0} \e^{(2-2aN)s}\e^{a(v(s;t_0,\alpha)-\e^{v(s;t_0,\alpha)})}\e^{v(s;t_0,\alpha)}(\e^{v(s;t_0,\alpha)}-1)\dd s.
\ee
Since $v<0$ and $v$ is a continuous function of $t_0, \alpha$, \eqref{4.9} tells us that $\eta$ depends continuously on $t_0, \alpha$. Our goal now is to show that there are $t_0, \alpha$ to make $\eta$ fulfill the condition \eqref{4.8}. We will use three steps to illustrate this fact.

{\em Step 1.} From \eqref{4.3} we get $v''>-\beta\e^{(2-2aN)t}\e^v$. Set $\omega=(2-2aN)t+v$, then
\be\label{4.10}
\omega''>-\beta\e^\omega.
\ee
However, \eqref{4.5} says that $v'(t)\geq 0$ for $t\leq t_0$, so we have $\omega'(t)>0$ for $t\leq t_0$. Multiplying the both sides of \eqref{4.10} by $\omega'$ and integrating on $(t, t_0)$, we get
\be\nn
(2-2aN)^2-(\omega'(t))^2>2\beta(\e^{\omega(t)}-\e^{(2-2aN)t_0-\alpha})>-2\beta\e^{(2-2aN)t_0-\alpha},~~t<t_0,
\ee
which means
\be\label{4.11}
0<v'(t;t_0,\alpha)<\sqrt{(2-2aN)^2+2\beta\e^{(2-2aN)t_0-\alpha}}-(2-2aN)\equiv K,~~t<t_0.
\ee
Furthermore, we obtain a useful inequality from \eqref{4.11},
\be\label{4.12}
-\alpha>v(t;t_0,\alpha)>-\alpha-K(t_0-t),~~t<t_0.
\ee

{\em Step 2.} Since $\e^{a(v-\e^v)}\e^v(\e^v-1)$ is a decreasing function in $v\in(-\infty, \ln\frac{1+a-\sqrt{1+a}}{a})$. Thus, when $\alpha\geq\ln\frac{a}{1+a-\sqrt{1+a}}$, in view of \eqref{4.3} and \eqref{4.12}, we have
\be\label{4.13}
v''<\beta\e^{(2-2aN)t}\e^{a(-\alpha-K(t_0-t)-\e^{-\alpha-K(t_0-t)})}\e^{-\alpha-K(t_0-t)}(\e^{-\alpha-K(t_0-t)}-1),~~t<t_0.
\ee
Integrating \eqref{4.13} over $(-\infty,t_0)$ gives
\bea\label{4.14}
\eta(t_0,\alpha)&=&\lim_{t\to -\infty}v'(t;t_0,\alpha)\nn\\
&\geq&-\beta\int_{-\infty}^{t_0}\e^{(2-2aN)s}\e^{a(-\alpha-K(t_0-s)-\e^{-\alpha-K(t_0-s)})}\e^{-\alpha-K(t_0-s)}(\e^{-\alpha-K(t_0-s)}-1)\dd s\nn\\
&=&\beta\e^{(2-2aN)t_0-a(\alpha+\e^{-\alpha})-\alpha}\left(\frac{1}{2-2aN+aK+K}-\frac{\e^{-\alpha}}{2-2aN+aK+2K}\right)\nn\\
&\geq&\beta\e^{(2-2aN)t_0-a(\alpha+\e^{-\alpha})-\alpha}\frac{1-\e^{-\alpha}}{2-2aN+aK+2K}\nn\\
&=&\frac{\beta\e^{(2-2aN)t_0-a(\alpha+\e^{-\alpha})-\alpha}(1-\e^{-\alpha})}{2-2aN+(a+2)(\sqrt{(2-2aN)^2+2\beta\e^{(2-2aN)t_0-\alpha}}-(2-2aN))}.
\eea

{\em Step 3.} By virtue of \eqref{4.11}, we have
\be\nn
0<\eta(t_0,\alpha)\leq\sqrt{(2-2aN)^2+2\beta\e^{(2-2aN)t_0-\alpha}}-(2-2aN).
\ee
Thus, for any $\alpha\geq\ln\frac{a}{1+a-\sqrt{1+a}}>0$, there is a suitable $t_0=t_0'$ so that $\eta(t_0',\alpha)<2N$. On the other hand, \eqref{4.14} tells us that for fixed $\alpha\geq\ln\frac{a}{1+a-\sqrt{1+a}}$, there exists some $t_0=t_0''$ to make $\eta(t_0'',\alpha)>2N$. As a consequence, we can find a point $t_0=t_0(\alpha)$ between $t_0'$ and $t_0''$ so that $\eta(t_0,\alpha)=2N$.
\end{proof}

In the following, we study the asymptotic behavior at infinity of the solution $v(t)$ of \eqref{4.3}--\eqref{4.4} derived from Lemma \ref{le4.3}.

\begin{lemma}\label{le4.4}
There is a constant $k>4+2N-4aN$ such that
\be\label{4.15}
\lim_{t\to\infty}v'(t)=-k.
\ee
\end{lemma}

\begin{proof}
From \eqref{4.3} and Lemma \ref{le4.2}, we see $v''<0$. Therefore, either $v'(t)\to -\infty$ or a finite number as $t\to\infty$. However, the first situation cannot happen. In fact, otherwise, if $v'(t)\to -\infty$ as $t\to\infty$, then there is a $\bar{t}$ such that $v'(t)<-3+2aN$(say) for $t>\bar{t}$. Thus, $v(t)<(2aN-3)t+C~(t>t_0)$ for some constant $C$ and
\bea\nn
\lim_{t\to\infty}v'(t)&=&\beta\int_{t_0}^{\infty} \e^{(2-2aN)s}\e^{a(v(s)-\e^{v(s)})}\e^{v(s)}(\e^{v(s)}-1)\dd s\nn\\
&>&-\beta\int_{t_0}^{\infty} \e^{(2-2aN)s}\e^{(2aN-3)s+C}\dd s=-\beta\e^{C-t_0}>-\infty,
\eea
which reaches a contradiction.

Therefore, there must be some $k>0$ to make \eqref{4.15} hold. All that remains is to show that $k>4+2N-4aN$.

Since $v'(t)$ is decreasing for $t\geq t_0$, then $v'(t)>-k, t\geq t_0$ and $v(t)>-kt+C, t\geq t_0$, where $C$ is a constant. Note that $v(t)\to -\infty$ and $v'(t)\to-k$ as $t\to\infty$, we see that the integral
\be\nn
\int_{t_0}^{\infty} \e^{(2-2aN)s}\e^{a(v(s)-\e^{v(s)})}\e^{v(s)}\dd s
\ee
converges. Hence, there must holds $k>\frac{2-2aN}{a+1}$, which in turn implies
\be\label{4.16}
\lim_{t\to\infty}\e^{(2-2aN)t}\e^{a(v(t)-\e^{v(t)})}\e^{v(t)}=\lim_{t\to\infty}\e^{t\left(2-2aN+a\frac{v(t)}{t}-a\frac{\e^{v(t)}}{t}+\frac{v(t)}{t}\right)}=0.
\ee

Now, multiplying the equation \eqref{4.3} by $v'$, integrating over $(-\infty, \infty)$, and using \eqref{4.16}, we have
\bea
\frac 12(k^2-4N^2)=&&\beta(2-2aN)\int_{-\infty}^{\infty}\e^{(2-2aN)s}\e^{a(v(s)-\e^{v(s)})}\e^{v(s)}(1-\e^{v(s)})\dd s\nn\\
&&+\frac{\beta(2-2aN)}{2}\int_{-\infty}^{\infty}\e^{(2-2aN)s}\e^{a(v(s)-\e^{v(s)})}\e^{2v(s)}\dd s\label{4.17}\\
&&+\beta a\int_{-\infty}^{\infty}\e^{(2-2aN)s}\e^{a(v(s)-\e^{v(s)})}(1-\e^{v(s)})\left(1-\frac{\e^{v(s)}}{2}\right)\e^{v(s)}v'(s)\dd s.\nn
\eea
On the right--hand side of \eqref{4.17}, the first term equals $(2-2aN)(k+2N)$, the second term is positive, and the last term
\bea\nn
&&\beta a\int_{-\infty}^{\infty}\e^{(2-2aN)s}\e^{a(v(s)-\e^{v(s)})}(1-\e^{v(s)})\left(1-\frac{\e^{v(s)}}{2}\right)\e^{v(s)}v'(s)\dd s.\nn\\
>&&\beta a\int_{-\infty}^{t_0}\e^{(2-2aN)s}\e^{a(v(s)-\e^{v(s)})}(1-\e^{v(s)})\left(1-\frac{\e^{v(s)}}{2}\right)\e^{v(s)}v'(s)\dd s\nn\\
&&+\beta a\e^{(2-2aN)t_0}\int_{t_0}^{\infty}\e^{a(v(s)-\e^{v(s)})}(1-\e^{v(s)})\left(1-\frac{\e^{v(s)}}{2}\right)\e^{v(s)}\dd v
\eea
is also positive.

In other words,
\be\label{4.18}
\frac 12(k^2-4N^2)>(2-2aN)(k+2N).
\ee
Therefore $k>4+2N-4aN$ as desired.

At this point, combining Lemmas \ref{le4.1}--\ref{le4.4}, we can complete the proof of Theorem \ref{th2.2}.
\end{proof}

\section{Existence of strings over compact surface}\label{s3}
\setcounter{equation}{0}

In this section we prove the existence of multiple strings over a compact surface $S$. It will be seen that the total number of strings is of technical significance. The proof splits into a few steps.

{\em Step 1. The perturbed problem}

Let $v_0$ be a solution of the equation
\be\label{5.1}
\triangle v_0=-\frac{4\pi N}{|S|}+4\pi\sum_{s=1}^N\delta_{p_s}(x).
\ee
Inserting $v=\vp+v_0$ in \eqref{1.21}, we obtain
\be\label{5.2}
\triangle\vp=\beta\e^{a(\vp+v_0-\e^{\vp+v_0})}\prod_{s=1}^N|x-p_s|^{-2a}\e^{\vp+v_0}(\e^{\vp+v_0}-1)+\frac{4\pi N}{|S|}~~~\text{on}~S.
\ee

Let $(U_s,(x^j))$ be an isothermal coordinate chart near $p_s\in S$ so that $x^j(p_s)=0~(j=1,2)$. According to \cite{Au}, when $U_s$ is small the function $v_0$ satisfies
\be\label{5.3}
v_0(x)=\ln|x|^2+w_s(x)~~~~\text{in}~(U_s,(x^j)),
\ee
where $w_s$ is a small function on $U_s$.

For any $\sigma>0$ small, let $\rho$ be a smooth function defined on $S$ so that $0\leq\rho\leq 1$ and
\be\nn
\rho(p)=
 \begin{cases}
  & 1,~~~|x(p)|<\sigma,\\
  & 0,~~~|x(p)|>2\sigma.
\end{cases}
\ee

In order to over the difficulty caused by the singular factor $\prod_{s=1}^N|x-p_s|^{-2a}$. We consider a perturbed version of the equation \eqref{5.2}
\be\label{5.5}
\triangle\vp=\beta\e^{a(\vp+v_0-\e^{\vp+v_0})}\prod_{s=1}^N(|x-p_s|^2+\delta\rho(x))^{-a}\e^{\vp+v_0}(\e^{\vp+v_0}-1)+\frac{4\pi N}{|S|}~~~\text{on}~S.
\ee

{\em Step 2. The sub/supersolution}

In the following we will obtain the solution of \eqref{5.5} via sub/supersolution.

We first find a subsolution of \eqref{5.5}.

\begin{lemma}\label{le5.1}
There is a smooth function $w_-$ on $S$ independent of $\delta$ such that
\be\label{5.6}
\triangle w_->\beta\e^{a(w_-+v_0-\e^{w_-+v_0})}\prod_{s=1}^N(|x-p_s|^2+\delta\rho(x))^{-a}\e^{w_-+v_0}(\e^{w_-+v_0}-1)+\frac{4\pi N}{|S|}
\ee
establishes on $S$ for some suitable $\beta$. In other words, $w_-$ is a subsolution of \eqref{5.5}.
\end{lemma}

\begin{proof}
Let $(U_s,(x^j))$ be a coordinate chart near $p_s\in S$ and $\sigma>0$ be a small number such that
\be\nn
\{x\in\mathbb{R}^2~|~|x|<3\sigma\}\subset\{x\in\mathbb{R}^2~|~x=x(p)~\text{for some}~p\in U_s\},~~s=1,2,\ldots,N.
\ee
Define a function $f_\sigma\in C^\infty(S)$ such that $0\leq f_\sigma\leq 1$ and
\be\nn
f_\sigma=
 \begin{cases}
  & 1,~~~|x(p)|<\sigma~\text{and}~p\in U_s, s=1,2,\ldots,N,\\
  & 0,~~~|x(p)|>2\sigma~\text{and}~p\in U_s, s=1,2,\ldots,N~\text{or}~p\in S-\cup_{s=1}^N U_s.
\end{cases}
\ee
If $C(\sigma)$ is a function satisfies
\be\label{5.7}
C(\sigma)=\frac{8\pi N}{|S|^2}\int_S f_\sigma \dd \Omega,
\ee
then the Theorem 4.7 in \cite{Au} says that the solution of the equation
\be\label{5.8}
\triangle w_-=\frac{8\pi N}{|S|}f_\sigma-C(\sigma)
\ee
is unique up to an additive constant. Besides, by \eqref{5.7} and the definition of $f_\sigma$, we see that $C(\sigma)\to 0$ as $\sigma\to 0$. Hence, we can choose a suitable $\sigma>0$ to make
\be\label{5.9}
\frac{8\pi N}{|S|}-C(\sigma)>\frac{4\pi N}{|S|}.
\ee

Set
\be\nn
U_s^\sigma=\{p\in S~|~p\in U_s~\text{and}~|x(p)|<\sigma\},~~s=1,2,\ldots,N.
\ee
Then, by virtue of \eqref{5.8} and \eqref{5.9}, we have
\be\nn
\triangle w_->\frac{4\pi N}{|S|}~~~\text{in}~\cup_{s=1}^N U_s^\sigma.
\ee

Furthermore, we can choose $w_-$ such that $\e^{w_-+v_0}-1<0$ on $S$. As a consequence, the inequality \eqref{5.6} is valid in $\cup_{s=1}^N U_s$ for any $\beta$ and $\delta$.

Denote
\bea\nn
\mu_0&=&\inf\left\{\e^{w_-+v_0}~|~x\in S-\cup_{s=1}^N U_s^\sigma\right\},\nn\\
\mu_1&=&\sup\left\{\e^{w_-+v_0}~|~x\in S-\cup_{s=1}^N U_s^\sigma\right\}.\nn
\eea
Then, we see $0<\mu_0<\mu_1<1$ and
\bea\nn
&&\e^{a(w_-+v_0-\e^{w_-+v_0})}\prod_{s=1}^N(|x-p_s|^2+\delta\rho(x))^{-a}\e^{w_-+v_0}(\e^{w_-+v_0}-1)\nn\\
\leq&&\mu_0^{a+1}(\mu_1-1)\e^{-a\mu_1}\prod_{s=1}^N(|x-p_s|^2+1)^{-a}<0~~~\text{in}~S-\cup_{s=1}^N U_s^\sigma.\nn
\eea
Hence, by making $\beta>0$ sufficiently large, we see that \eqref{5.6} holds on $S-\cup_{s=1}^N U_s^\sigma$ as well. Then, Lemma \ref{le5.1} follows.
\end{proof}

\begin{lemma}\label{le5.2}
Define $\vp_1=-v_0$, then $\vp_1>w_-$ on $S$.
\end{lemma}

\begin{proof}
Let $\sigma>0$ is small so that $U_s^\sigma\cap U_{s'}^\sigma=\emptyset$ for $s\neq s'$. Clearly, $w_-<\vp_1$ in $\overline{U}_s^\sigma (s=1,\ldots,N)$ when $\sigma$ is sufficiently small.

By virtue of \eqref{5.1}, we can rewrite \eqref{5.6} in the form
\be\label{5.10}
\triangle (w_-+v_0)>\beta\e^{a(w_-+v_0-\e^{w_-+v_0})}\prod_{s=1}^N(|x-p_s|^2+\delta\rho(x))^{-a}\e^{w_-+v_0}(\e^{w_-+v_0}-1)~~~\text{in}~S-\cup_{s=1}^N U_s^\sigma.
\ee
We already have $w_-+v_0<0$ on $\pa U_s^\sigma~(s=1,\ldots,N)$. If there exists a point $p$ such that $(w_-+v_0)(p)\geq 0$, then the function $w_-+v_0$ has a nonnegative interior maximum in $S-\cup_{s=1}^N U_s^\sigma$ which is impossible due to \eqref{5.10} and the maximum principle.

In summary, $\vp_1>w_-$ on $S$.
\end{proof}

{\em Step 3. Solution of the perturbed equation}

It is easy to see that $\vp_1$ is a supersolution of \eqref{5.5} and $\vp_1$ is singular at the points $p_1,\ldots,p_N$. Define the following iterative scheme
\bea\label{5.11}
(\triangle-C_\delta)\vp_n=&&\beta\e^{a(\vp_{n-1}+v_0-\e^{\vp_{n-1}+v_0})}\prod_{s=1}^N(|x-p_s|^2+\delta\rho(x))^{-a}\e^{\vp_{n-1}+v_0}(\e^{\vp_{n-1}+v_0}-1)\nn\\
&&-C_\delta\vp_{n-1}+\frac{4\pi N}{|S|}~~~\text{on}~S,~~n=2,3\ldots,
\eea
where $C_\delta>0$ is a constant to be determined in the following.

Set
\be\label{5.12}
f(t)=\e^{a(t-\e^t)}\e^t(\e^t-1),
\ee
then $f'(t)=\e^{a(t-\e^t)+2t}(4-\e^t)$ is bounded for $t\in\mathbb{R}$. Let
\be\nn
C_\delta=1+\beta\sup_{x\in S}\prod_{s=1}^N(|x-p_s|^2+\delta\rho(x))^{-a}\cdot\sup_{t\in\mathbb{R}}\{f'(t)\}
\ee
in \eqref{5.11}. We have

\begin{lemma}\label{le5.3}
There holds on $S$ the inequality
\be\label{5.13}
\vp_1>\vp_2>\ldots>\vp_n>\ldots>w_-.
\ee
\end{lemma}

\begin{proof}
Lemma \ref{le5.2} says that $\vp_1>w_-$. By \eqref{5.11}, we see that
\be\label{5.14}
(\triangle-C_\delta)\vp_2=-C_\delta\vp_1+\frac{4\pi N}{|S|}~~~\text{on}~S.
\ee
Thus $(\triangle-C_\delta)(\vp_2-\vp_1)=0$ in $S-\{p_1,\ldots,p_N\}$. Since $\vp_1\in L^p(S)$ for any $p>1$, we get that $\vp_2\in W^{2,p}(S)$ and $\vp_2\in C^{1,\alpha}(S)$ for any $0<\alpha<1$. In particular, $\vp_2$ is bounded. Using the maximum principle, we get $\vp_1>\vp_2$.

Besides, using the notation \eqref{5.12}, the equality \eqref{5.6} can be rewritten as
\be\nn
\triangle w_->\beta f(w_-+v_0)\prod_{s=1}^N(|x-p_s|^2+\delta\rho(x))^{-a}+\frac{4\pi N}{|S|}.
\ee
Thus, in view of Lemma \ref{le5.2} and \eqref{5.14}, we have
\bea\nn
&&(\triangle-C_\delta)(w_--\vp_2)\nn\\
>&&\beta\prod_{s=1}^N(|x-p_s|^2+\delta\rho(x))^{-a}(f(w_-+v_0)-f(\vp_1+v_0))-C_\delta(w_--\vp_1)\nn\\
=&&\left(\beta\prod_{s=1}^N(|x-p_s|^2+\delta\rho(x))^{-a}f'(v_0+\xi)-C_\delta\right)(w_--\vp_1)\nn\\
>&&0~~~~(w_-<\xi<\vp_1).\nn
\eea
Hence, the maximum principle gives us $w_-<\vp_2$.

Suppose we have shown that $\vp_{k-1}>\vp_k>w_-$ with $k\geq 2$. Then, for some $\xi_k$ lying between $\vp_{k-1}$ and $\vp_k$, we have
\bea\nn
&&(\triangle-C_\delta)(\vp_{k+1}-\vp_k)\nn\\
>&&\beta\prod_{s=1}^N(|x-p_s|^2+\delta\rho(x))^{-a}(f(\vp_k+v_0)-f(\vp_{k-1}+v_0))-C_\delta(\vp_k-\vp_{k-1})\nn\\
=&&\left(\beta\prod_{s=1}^N(|x-p_s|^2+\delta\rho(x))^{-a}f'(v_0+\xi_k)-C_\delta\right)(\vp_k-\vp_{k-1})\nn\\
>&&0.\nn
\eea
Hence, $\vp_{k+1}<\vp_k$.

Moreover, for some $\xi_k$ lying between $w_-$ and $\vp_k$, we have
\bea\nn
&&(\triangle-C_\delta)(w_--\vp_{k+1})\nn\\
>&&\beta\prod_{s=1}^N(|x-p_s|^2+\delta\rho(x))^{-a}(f(w_-+v_0)-f(\vp_k+v_0))-C_\delta(w_--\vp_k)\nn\\
=&&\left(\beta\prod_{s=1}^N(|x-p_s|^2+\delta\rho(x))^{-a}f'(v_0+\xi)-C_\delta\right)(w_--\vp_k)\nn\\
>&&0.\nn
\eea
Thus, again, $w_-<\vp_{k+1}$. Consequently, the general property \eqref{5.13} holds.
\end{proof}

The equality \eqref{5.13} implies that the sequence $\{\vp_n\}$ defined in \eqref{5.11} is convergent, we may assume
\be\nn
\lim_{n\to\infty}\vp_n\equiv \vp^\delta.
\ee
Then $\vp^\delta$ is a solution of the equation \eqref{5.5}, and satisfying
\be\label{5.15}
\vp_1>\vp^\delta\geq w_-~~~\text{on}~S.
\ee

In order to find a solution of the original equation \eqref{5.2}, we need to consider the $\delta\to 0$ limit.

{\em Step 4. Passage to limit}

We see from \eqref{5.15} that, for any $p>1$, there exists a constant $C>0$ independent of $\delta$ so that
\be\label{5.16}
||\vp^\delta||_{L^p(S)}\leq C.
\ee
Next, we observe that the one--parameter function $f(\vp^\delta+v_0)$ is uniformly bounded.

Besides, since $(|x-p_s|^2+\delta\rho(x))^{-a}\leq(|x-p_s|^2)^{-a}$, we see that there is a constant $C>0$ independent of $\delta$ so that
\be\label{5.17}
|||x-p_s|^2+\delta\rho(x)||_{L^p(S)}\leq C
\ee
provided that
\be\label{5.18}
2ap<2.
\ee
Combining with the condition $\chi(S)=aN$ and $\chi(S)=2$, we arrive at
\be\label{5.19}
p<\frac{N}{2}.
\ee

Therefore, if $N\geq 3$, then there is a $p>1$ which makes \eqref{5.19} valid. Using \eqref{5.16} and \eqref{5.17} in \eqref{5.5}, we see that there is a constant $C>0$ independent of $\delta$ to confirm the $W^{2,p}$--norm of $\vp^\delta$,
\be\label{5.20}
||\vp^\delta||_{W^{2,p}(S)}\leq C,~~~\forall~0<\delta<1.
\ee
Furthermore, using the continuous embedding $W^{k,p}(S)\to C^m(S)$ for $0\leq m<k-2/p$ with $k=2$ and $p>1$, we see that $\{\vp^\delta\}$ is $C(S)$ bounded for any $\delta$.

Now, we rewrite \eqref{5.5} with $\vp=\vp^\delta$ as
\be\label{5.21}
\triangle\vp^\delta=\beta\e^{a(\vp^\delta+v_0-\e^{\vp^\delta+v_0})}\prod_{s=1}^N(|x-p_s|^2+\delta\rho(x))^{-a}\e^{\vp^\delta+v_0}(\e^{\vp^\delta+v_0}-1)+\frac{4\pi N}{|S|}.
\ee
Clearly, the right--hand side of \eqref{5.21} has uniformly bounded $L^p$--norm for any $p>1$. Thus \eqref{5.20} establishes for any $p>1$. Therefore, $\{\vp^\delta\}$ is bounded in $C^{1,\alpha}$ for any $0<\alpha<1$. In view of \eqref{5.21}, we obtain that $\{\vp^\delta\}$ is bounded in $C^{2,\alpha}(S)$. From the compact embedding $C^{2,\alpha}(S)\to C^{2}(S)$, we know that there is a convergent subsequence $\{\vp^{\delta_n}\} (\delta_n\to 0~\text{as}~n\to\infty)$ such that
\be\label{5.22}
\vp^{\delta_n}\to~\text{some element}~\vp~~~\text{in}~C^2(S).
\ee
Inserting this fact into \eqref{5.21}, we see that $\vp$ is a solution of \eqref{5.2}.

The proof of Theorem \ref{th2.3} is complete.

\end{document}